\documentclass[aps,12pt,superscriptaddress,amsfonts,amssymb,amsmath]{revtex4-2}
\pdfoutput=1

\usepackage[a4paper,margin=1in]{geometry}
\usepackage[T1]{fontenc}
\usepackage[utf8]{inputenc}
\usepackage{lmodern}
\usepackage{amsmath,amssymb,amsthm,mathtools}
\usepackage{physics}
\usepackage{bm}
\usepackage{microtype}
\usepackage{hyperref}
\usepackage{enumitem}
\usepackage{booktabs}
\usepackage{xcolor}
\usepackage{graphicx}
\hypersetup{colorlinks=true,linkcolor=blue,citecolor=blue,urlcolor=blue}

\newtheorem*{proposition*}{Proposition}
\newtheorem*{remark*}{Remark}

\begin{document}

\title{Effective Observer-Split Source Terms in Rotating Frames and Gravitomagnetic Backgrounds in Extended Aharonov-Bohm Electrodynamics}

\author{A.\ Iadicicco }
\affiliation{Università Popolare ``Nikola Tesla'', Via Trinità 33, 83100 Avellino, Italy}

\author{G.\ Modanese \footnote{Email address: giovanni.modanese@unibz.it}}
\affiliation{Libera Universita di Bozen-Bolzano, Facoltà di Ingegneria, Via Buozzi 1, I-39100 Bolzano, Italy}

\author{L.\ Verolino}
\affiliation{Università di Napoli ``Federico II'', Dipartimento di Ingegneria Elettrica e delle Tecnologie dell'Informazione, 
Via Claudio 21, 80125 Napoli, Italy}


\newtheorem{proposition}{Proposition}
\newtheorem{corollary}{Corollary}
\newtheorem{remark}{Remark}
\newcommand{\GGM}{\mathcal G_{\mathrm{GM}}}

\begin{abstract}
{\color{black}
We examine whether the observer split of a covariantly conserved electric current in a rotating finite system can provide a physically constrained effective source for the scalar sector of extended Aharonov--Bohm electrodynamics. We first recall that rotation, a stationary gravitomagnetic background, or a change of coordinates cannot by themselves generate genuine microscopic charge non-conservation for standard locally $U(1)$-invariant matter; this simple no-go statement is used to delimit the phenomenological. In a $3+1$ decomposition adapted to a rotating apparatus, the transport continuity equation contains the exact observer-split term $I_{\mathrm{split}}=N^{-1}D_i(\rho\beta^i)$, which reduces for weak rigid rotation to $I_G=\Omega\partial_\phi\rho$. We propose to use this restricted, congruence-dependent source as a phenomenological input to the extended scalar sector. The relevant observer/evolution structure is not an arbitrary coordinate choice: the laboratory synchronization and readout protocol select the foliation, while the rotor motion fixes the evolution field relative to it. We characterize which features of a genuine extra-current can and cannot be represented by this construction, and discuss falsifiable signatures, including rotation reversal, axisymmetric-source suppression, gradient scaling, and a recently proposed optically gated rotating-disc test. Analogue-gravity systems are also identified as a possible future setting in which a preferred material flow naturally selects the observer congruence.
}
\end{abstract}

\maketitle

\section{Introduction}

Extended Aharonov-Bohm electrodynamics allows the scalar quantity $S = \partial_\mu A^\mu$ to acquire dynamical significance and can remain consistent even when the primary source current is not locally conserved, because a secondary non-local current restores conservation at the level of the observable total source \cite{Modanese2017}. In the usual applications of this framework, the starting point is therefore a genuinely non-conserved primary current.
The same technique has been explored for coherent non-conservative tunnelling, fractional quantum mechanics, gauge-wave propagation, and more recently in gravitationally coupled scalar-tensor models \cite{Modanese2017,Modanese2017dipole,MinottiModanese2023,Minotti2024,MinottiModanese2026}.

{\color{black}
The broader motivation for this extension is to provide a consistent field-theoretic description when an effective or microscopic primary current does not obey a local continuity equation. If the primary source has an extra-current $I=\partial_\mu j^\mu\neq0$, the field equations generate a non-local secondary current whose divergence compensates that of the primary current, so that the total source of the electromagnetic field remains conserved. This ``charge-conservation censorship'' mechanism was introduced in Ref.~\cite{Modanese2017}. Subsequent work examined explicit non-local quantum sources, radiative multipoles, the generation and detection of gauge waves, and couplings to bosonic matter and scalar--tensor gravity \cite{Modanese2017dipole,MinottiModanese2021,MinottiModanese2023,Minotti2024,MinottiModanese2026}. More recently, idealized wave-function reduction and biased tunnelling have provided another concrete class of source models in which local charge conservation may fail at the event level and for which experimental tests with electronic devices have been proposed \cite{MinottiModaneseCollapse2026}. The present rotating construction explores whether a macroscopically defined observer-split transport term can play an analogous source role at the phenomenological level, without attributing a microscopic charge anomaly to rotation itself.

We also stress a point of terminology that is relevant when comparing this framework with the experimentally established Aharonov--Bohm interference effect. The present theory is not proposed as an alternative explanation of the standard magnetic Aharonov--Bohm phase shift \cite{AharonovBohm1959}. The name ``Aharonov--Bohm electrodynamics'' refers here to the reduced-gauge line of development associated historically with Ohmura's formulation and with the later discussion by Aharonov and Bohm of the physical status of the electromagnetic potentials and of subsidiary conditions \cite{Ohmura1956,AharonovBohm1963}. In the sector in which the source current is locally conserved, the secondary current vanishes and the field equations reduce to Maxwell electrodynamics. Thus the well-tested ordinary AB effect remains part of the standard electromagnetic sector; the additional scalar dynamics becomes relevant only when a non-conserved primary source, or the effective surrogate studied here, is introduced.
}

{\color{black}The question addressed here is therefore more specific: can rotation, or more generally the shift metric structure associated with a rotating frame or a stationary gravitomagnetic background, generate an \emph{effective} source term of the kind required by the scalar sector, even when the underlying charged matter is otherwise standard and its microscopic current remains conserved?}

The answer is twofold. At the microscopic level, the answer is no. For standard matter minimally coupled to electromagnetism and gravity through a generally covariant and locally $U(1)$-invariant action, the physical four-current is covariantly conserved on shell. Hence neither a stationary gravitomagnetic field nor a mere change of coordinates can by themselves produce a genuine microscopic charge anomaly.

At the level of observer-adapted transport variables, however, a weaker statement holds. Once the conserved four-current is decomposed with respect to a rotating $3+1$ slicing, the continuity equation for the observer-measured transport current on the spatial slice acquires an exact source term proportional to the shift vector. This source term is not a new microscopic law. Rather, it is the term required when the covariant conservation law is rewritten in variables adapted to the rotating frame of the device.

Within the gravito-electromagnetic description based on observer splitting, rotation and Sagnac-type effects in stationary geometries can be formulated through a formal gravito-magnetic Aharonov-Bohm analogy \cite{Ruggiero2004,Ruggiero2005,Ruggiero2010}. This is useful here only at the geometric level: the same shift structure that governs Sagnac-type holonomy also appears in the observer-split current decomposition used below, although no microscopic source law follows from that analogy by itself. The same observer-splitting formalism is widely used to describe relative dynamics and propagation effects in rotating or stationary spacetimes.

The distinction is central to the present paper. The claim is therefore deliberately modest and precise. We do not derive a new microscopic violation of charge conservation from rotation. Rather, we identify an observer-dependent source term that emerges in the rotating $3+1$ description of finite devices, and we investigate whether that term can consistently serve as the input of a phenomenological AB-type scalar-sector closure. In this sense the resulting bridge is effective rather than fundamental, frame-tied rather than local-inertial, and operationally meaningful only when the measurement protocol is adapted to the rotating apparatus, for example involving macroscopic wavefunctions and their long-range correlations.

\section{Extended AB electrodynamics in flat and curved spacetime}
In the covariant formulation of extended AB electrodynamics, the Lagrangian contains an additional term proportional to $(\partial_\mu A^\mu)^2$. One may define
\begin{equation}
S\equiv \partial_\mu A^\mu,
\end{equation}
and derive the generalized field equations in flat spacetime as \cite{Modanese2017}
\begin{equation}
\partial_\mu F^{\mu\nu}=j^\nu+i^\nu,
\qquad
i^\nu=-\partial^\nu\Box^{-1}(\partial_\alpha j^\alpha).
\label{eq:ABflat}
\end{equation}
Here $\Box\equiv \partial_\mu\partial^\mu$ denotes the flat-space d'Alembert operator. The observable current $j^\nu+i^\nu$ is always conserved, even when the primary current $j^\nu$ is not.

As a minimal covariant ansatz on a background spacetime $(M,g_{\mu\nu})$, one may write
\begin{equation}
\nabla_\mu F^{\mu\nu}=j^\nu+i^\nu,
\qquad
i^\nu=-\nabla^\nu\Box_g^{-1}(\nabla_\alpha j^\alpha),
\label{eq:ABcurved}
\end{equation}
where $\Box_g\equiv \nabla_\mu\nabla^\mu$ is the covariant d'Alembert operator. Formally, assuming $\Box_g^{-1}$ is defined with appropriate boundary conditions, one then has
\begin{equation}
\nabla_\nu\qty(j^\nu+i^\nu)=0
\end{equation}
identically. Equation \eqref{eq:ABcurved} should therefore be read as a minimal covariant extension of the flat-space secondary current formula, consistent with recent work coupling AB electrodynamics to scalar-tensor gravity \cite{Minotti2024,MinottiModanese2026}.

If the physical current obeys $\nabla_\mu j^\mu=0$, then $i^\mu=0$ and the theory collapses back to Maxwell theory.

\section{No-go premise: gravitomagnetism alone does not create a true extra-current}

{\color{black}
The proposition below is  essentially the curved-spacetime Noether identity for standard charged matter. Its purpose here is to make explicit what rotation alone can and cannot do. For a standard microscopic Noether current one has $\nabla_\mu J^\mu=0$, even in a rotating description or in a stationary gravitomagnetic background. The nonzero quantity considered later in the paper is instead the observer-split transport term obtained by decomposing this same conserved current. Thus the no-go result prevents us from interpreting that term as microscopic charge creation or destruction; the additional step of coupling it to the AB scalar sector is a separate phenomenological hypothesis.
}

\begin{proposition*}[No-go for standard matter]
Let $\psi$ denote any standard matter field minimally coupled to electromagnetism and gravity through a generally covariant, locally $U(1)$-invariant action $S_m[\psi,A,g]$ on a stationary spacetime with nonzero gravitomagnetic components $g_{0i}$. If the matter equations of motion are satisfied, then the electric current obtained from
\begin{equation}
J^\mu\equiv \frac{1}{\sqrt{-g}}\frac{\delta S_m}{\delta A_\mu}
\end{equation}
obeys
\begin{equation}
\nabla_\mu J^\mu=0.
\end{equation}
Therefore a stationary gravitomagnetic background generated by rotating masses does not, by itself, generate a genuine microscopic source for the scalar sector of extended AB electrodynamics.
\end{proposition*}

\emph{Proof sketch.}
The statement is essentially Noether's theorem in curved spacetime. Under an infinitesimal local phase transformation $\psi\to e^{iq\alpha}\psi$, $A_\mu\to A_\mu+\nabla_\mu\alpha$, local gauge invariance of $S_m$ implies, on shell,
\begin{equation}
0=\delta S_m = \int d^4x\,\sqrt{-g}\, (\nabla_\mu J^\mu)\,\alpha,
\end{equation}
for arbitrary $\alpha$, hence $\nabla_\mu J^\mu=0$.
\hfill $\square$

\begin{remark*}
This result is stronger than a coordinate statement. In a curved or rotating coordinate system one may have $\partial_\mu J^\mu\neq 0$ while $\nabla_\mu J^\mu=0$, but this is a coordinate artifact. A covariant formulation such as \eqref{eq:ABcurved} uses $\nabla_\mu J^\mu$, not $\partial_\mu J^\mu$.
\end{remark*}

Accordingly, any interpretation in terms of genuine microscopic charge non-conservation requires additional assumptions beyond standard general relativity coupled to standard charged matter.

\section{The observer split in a rotating frame}

\subsection{ADM split and observer-measured transport variables}

Consider the ADM form of a stationary metric,
$$
ds^2 = -N^2 dt^2 + h_{ij}(dx^i + \beta^i dt)(dx^j + \beta^j dt),
$$
where $N$ is the lapse and $\beta^i$ is the shift vector. {\color{black}We use relativistic units with $c=1$ throughout the $3+1$ split. If physical time and length units are retained, $N$ and $\beta^i$ have dimensions of velocity and $N\to c$ in the Minkowski limit; the projected continuity equation is rescaled consistently.} In stationary weak-field geometries with frame dragging, $\beta^i$ is proportional to the gravitomagnetic potential.

Let $n^\mu$ be the future-directed unit normal to the hypersurfaces $t = \mathrm{const}$, with $n_\mu n^\mu = -1$. {\color{black}It is useful to distinguish this normal field from the ADM evolution vector $t^\mu\equiv(\partial/\partial t)^\mu$, which decomposes geometrically as
\begin{equation}
    t^\mu=Nn^\mu+\beta^\mu, \qquad n_\mu\beta^\mu=0.
    \label{eq:evolutionvector}
\end{equation}
Thus $n^\mu$ specifies the Eulerian observer congruence normal to the slices, whereas $t^\mu$ specifies how the spatial coordinates are carried from one slice to the next; the shift $\beta^\mu$ is the spatial part of this evolution field.}
The physical four-current $J^\mu$ can be decomposed as
$$
J^\mu = \rho\, n^\mu + j^\mu,
$$
with
$$
\rho \equiv -n_\mu J^\mu, \qquad j^\mu \equiv h^\mu{}_{\nu} J^\nu, \qquad n_\mu j^\mu = 0.
$$
Here $\rho$ is the charge density measured by the Eulerian observers associated with the chosen slicing, while $j^\mu$ is the transport current measured relative to that observer congruence. This observer split should be interpreted carefully. The fundamental source is still the covariantly conserved four-current $J^\mu$. The quantity $j^\mu$ is not a replacement for $J^\mu$; it is the transport current measured on the spatial slice defined by the rotating frame.

In particular, “spatial” means orthogonal to $n^\mu$, not necessarily with vanishing time component in arbitrary coordinates. In a slicing with nonzero shift, the coordinate time direction does not coincide with the physical time direction of the Eulerian observers, so a vector that is spatial in the observer sense need not have trivial coordinate components in an arbitrary chart.

{\color{black}
\subsection{Physical selection of the observer congruence}
\label{sec:physicalcongruence}
A central issue for any observer-dependent construction is what, if anything, selects the relevant observer structure. In the present model the foliation is not meant to be freely adjustable after the phenomenological closure has been specified. More precisely, the material world-lines of the rotor do not by themselves determine a unique ADM foliation. In a laboratory realization, the synchronization and readout protocol select the hypersurface normal $n^\mu$, while the rotor phase and angular motion specify the evolution field relative to those hypersurfaces and hence the shift $\beta^\mu$. The phenomenological source is therefore relational: it depends on the specified laboratory observer structure together with the rotating source motion.

This distinction separates dependence on physically supplied observer data from arbitrary coordinate dependence. A passive change of coordinates describing the same metric, current, foliation normal, evolution field, and apparatus cannot change a predicted observable. By contrast, changing the physically supplied foliation or evolution field changes the transport variables and therefore changes $I_{\rm split}$. Such dependence is familiar in effective descriptions of matter, where a medium four-velocity or detector world-line forms part of the physical specification. The present proposal consequently does not assign physical significance to the split source associated with every possible foliation; it assigns phenomenological significance only after the observer/evolution structure has been fixed operationally by the apparatus and measurement protocol.

This observation also suggests a possible connection with analogue-gravity systems. In acoustic or optical analogue spacetimes a material flow or moving medium naturally selects a preferred laboratory congruence, while perturbations propagate according to an effective curved metric \cite{BarceloLiberatiVisser2005}. Such systems do not reproduce the gravitational field equations themselves, but they may provide an experimentally accessible setting in which an observer split with a physically fixed shift-like flow can be implemented and tested. We leave the construction of a specific analogue-gravity realization to future work.
}

\subsection{Projected continuity equation and exact split source}

Starting from covariant conservation,
$$
\nabla_\mu J^\mu = \frac{1}{\sqrt{-g}}\,\partial_\mu\!\left(\sqrt{-g}\,J^\mu\right)=0,
$$
with
$$
\sqrt{-g} = N\sqrt{h},
$$
and using the above decomposition, one finds
$$
\partial_t(\sqrt{h}\,\rho) + \partial_i(N\sqrt{h}\,j^i) = \partial_i(\sqrt{h}\,\rho\,\beta^i).
$$
Equivalently,
\begin{equation}
\frac{1}{\sqrt{h}}\partial_t(\sqrt{h}\,\rho)+D_i(Nj^i)=D_i(\rho\,\beta^i).
\label{eq:projected_general}
\end{equation}
For the stationary slicing considered here, $\partial_t h_{ij}=0$ and therefore $\partial_t\sqrt{h}=0$, so Eq.~\eqref{eq:projected_general} reduces to
\begin{equation}
\partial_t\rho+D_i(Nj^i)=D_i(\rho\,\beta^i).
\label{eq:projected_stationary}
\end{equation}
Here the spatial divergence is naturally written with the Levi-Civita derivative $D_i$ of the spatial metric $h_{ij}$, whereas the time term reduces to an ordinary derivative precisely because the slicing is stationary. This equation is the key identity of the paper. It does not state that charge is locally created or destroyed. Rather, it states that when the exact covariant conservation law is rewritten in terms of observer-measured transport variables on a slicing with nonzero shift, the continuity equation for the transport current is not closed by itself: an additional term appears on the right-hand side. We define the exact split source
\begin{equation}
I_{\mathrm{split}} \equiv \frac{1}{N} D_i(\rho\,\beta^i).
        \label{eq:defIsplit}
\end{equation}
This term is observer-dependent and should not be interpreted as a frame-independent microscopic anomaly. It is the bookkeeping term required when the conserved four-current is rewritten in transport variables adapted to a rotating slicing.
 
\subsection{Weak-field rigid rotation and transient source condition}

In the weak-field, slow-rotation regime one has $N \simeq 1$, so the exact split source reduces to
\begin{equation}
    I_G = D_i(\rho\,\beta^i).
\label{eq:IGdef}
\end{equation}

For a rigidly rotating frame in flat space,
$$
\boldsymbol{\beta} = \boldsymbol{\Omega} \times \mathbf{r},
$$
and since $\nabla\cdot(\boldsymbol{\Omega} \times \mathbf{r}) = 0$, one obtains
$$
I_G = \nabla\cdot\big[\rho(\boldsymbol{\Omega} \times \mathbf{r})\big] = (\boldsymbol{\Omega} \times \mathbf{r})\cdot \nabla \rho.
$$
Hence the split source vanishes for spatially uniform charge density and is controlled by charge inhomogeneity in the rotating frame.

For applications it is useful to decompose the density as
$$
\rho(r,\phi,z,t) = \rho_0(r,z,t) + \delta\rho_s(r,\phi,z,t),
$$
where $\rho_0$ is approximately axisymmetric and $\delta\rho_s$ is a localized transient perturbation. In a frame rotating about the $z$ axis one then has
\begin{equation}
    I_G = \Omega\,\partial_\phi \rho \simeq \Omega\,\partial_\phi(\delta\rho_s)
\label{eq:IGphi}
\end{equation}
whenever the background term is axisymmetric to leading order.

The activation criterion is therefore not the existence of a transient as such, but the presence of a transient non-axisymmetric angular charge gradient in the rotating frame.

\section{Phenomenological AB closure sourced by the split term}

Up to this point, no modification of standard Maxwell theory has been made. The physical microscopic source remains the covariantly conserved four-current $J^\mu$, and the split source derived above is simply an exact rewriting of that conservation law in observer-adapted transport variables.

We now introduce an additional phenomenological step. Motivated by the exact source term appearing in the rotating-frame continuity equation, we define an observer-split auxiliary current. {\color{black}A coordinate-independent way of writing the current that in ADM components is $J_{\mathrm{eff}}^\mu=(\rho/N,j^i)$ is
\begin{equation}
    J_{\mathrm{eff}}^\mu
    \equiv J^\mu+\frac{\rho}{N}\,\beta^\mu
    =\frac{\rho}{N}t^\mu+j^\mu.
    \label{eq:Jeffcov}
\end{equation}
The equality follows directly from $J^\mu=\rho n^\mu+j^\mu$ and Eq.~\eqref{eq:evolutionvector}. It also makes clear that $J_{\mathrm{eff}}^\mu$ is a genuine four-vector once the foliation normal $n^\mu$ and evolution field $t^\mu$ have been supplied; its observer dependence is therefore not a passive coordinate dependence.

The divergence can be obtained without choosing ADM coordinates. For any vector $V^\mu$ tangent to the spatial hypersurfaces, one has
\begin{equation}
    \nabla_\mu V^\mu=D_\mu V^\mu+a_\mu V^\mu,
    \qquad
    a_\mu\equiv n^\nu\nabla_\nu n_\mu=D_\mu\ln N,
    \label{eq:spatialdividentity}
\end{equation}
where $D_\mu$ is the Levi--Civita derivative intrinsic to the hypersurfaces. Setting $V^\mu=(\rho/N)\beta^\mu$ gives
\begin{align}
\nabla_\mu\!\left(\frac{\rho}{N}\beta^\mu\right)
&=D_\mu\!\left(\frac{\rho}{N}\beta^\mu\right)
 +\frac{\rho}{N}\beta^\mu D_\mu\ln N \\
&=\frac{1}{N}D_\mu(\rho\beta^\mu).
\label{eq:Jeffdivstep}
\end{align}
Since the physical current obeys $\nabla_\mu J^\mu=0$, it follows that
\begin{equation}
    \nabla_\mu J_{\mathrm{eff}}^\mu
    =\frac{1}{N}D_\mu(\rho\beta^\mu)
    =I_{\mathrm{split}}.
    \label{eq:Jeffdivcov}
\end{equation}
Thus the same quantity previously obtained from the projected continuity equation is the four-divergence of a geometrically defined auxiliary current, once the observer/evolution structure is fixed. This does not make $I_{\mathrm{split}}$ observer independent: changing $n^\mu$ or $t^\mu$ changes the constructed current. It does show, however, that after those physical data are specified the construction is covariant under passive coordinate transformations.}

From this point onward we are no longer merely rewriting standard Maxwell theory in observer-adapted variables. The additional hypothesis is deliberately restricted to the scalar sector: we postulate that the exact observer-split divergence can act as the source entering the extended AB scalar equation,
\begin{equation}
    \Box_g S = \mu_0 I_{\mathrm{split}}.
    \label{eq:boxg}
\end{equation}
This is the phenomenological closure to be tested; it does not assert that the auxiliary current $J_{\mathrm{eff}}^\mu$ itself replaces the physical Maxwell source.

{\color{black}
This distinction is useful because $J_{\mathrm{eff}}^\mu$ is not uniquely determined by its divergence: adding any conserved vector field would leave $I_{\mathrm{split}}$ unchanged but would in general alter the ordinary transverse Maxwell source. To isolate the scalar hypothesis from such an additional and unnecessary assumption, one may introduce a minimal longitudinal surrogate primary current
\begin{equation}
    \mathcal J^\mu
    \equiv J^\mu + \nabla^\mu\Box_g^{-1} I_{\mathrm{split}},
    \qquad
    \nabla_\mu\mathcal J^\mu=I_{\mathrm{split}},
    \label{eq:Jsurrogate}
\end{equation}
where the Green function defining $\Box_g^{-1}$ is fixed by the chosen boundary conditions. The associated extended-AB secondary current is then
\begin{equation}
    i_{\mathrm{split}}^\mu
    =-\nabla^\mu\Box_g^{-1} I_{\mathrm{split}},
    \label{eq:isplitlong}
\end{equation}
so that
\begin{equation}
    \mathcal J^\mu+i_{\mathrm{split}}^\mu=J^\mu,
    \qquad
    \nabla_\mu\bigl(\mathcal J^\mu+i_{\mathrm{split}}^\mu\bigr)=0.
    \label{eq:totalsourceJ}
\end{equation}
Thus the minimal closure leaves the ordinary Maxwell source $J^\mu$ unchanged while allowing a non-trivial scalar/potential sector sourced by $I_{\mathrm{split}}$. Equations~\eqref{eq:Jsurrogate}--\eqref{eq:totalsourceJ} are not an additional microscopic model for matter; they are a convenient longitudinal realization of the surrogate source required by Eq.~\eqref{eq:boxg}. The geometrically constructed current $J_{\mathrm{eff}}^\mu$ remains useful because Eq.~\eqref{eq:Jeffdivcov} proves that $I_{\mathrm{split}}$ is the divergence of a well-defined observer-dependent four-vector, but no physical significance is assigned here to the divergence-free part of $J_{\mathrm{eff}}^\mu$.
}

The closure should therefore be read as an effective macroscopic hypothesis for rotating finite systems. The microscopic Noether current remains covariantly conserved, and the standard Maxwell field sourced by that current is not altered by the minimal longitudinal construction above.

{\color{black}
\section{Which features of a genuine extra-current can the split model represent?}
\label{sec:surrogate}
The observer-split term is not a generic substitute for a genuine frame-independent extra-current $I=\nabla_\mu J^\mu$. It represents only a restricted class of source structures fixed by the geometry and by the transported charge distribution. This restriction is important both conceptually and experimentally.

In the weak-field rigid-rotation limit,
\begin{equation}
I_{\rm split}\simeq I_G=\Omega\,\partial_\phi\rho .
\end{equation}
Consequently the effective source has several characteristic properties. It is odd under reversal of the angular velocity, $\Omega\to-\Omega$; it vanishes for an exactly axisymmetric density; it is controlled by spatial gradients rather than by a uniform density offset; and on a closed ring its positive and negative angular contributions balance,
\begin{equation}
\int_0^{2\pi} I_G\,d\phi =0,
\end{equation}
for a single-valued periodic density. More generally, because $I_{\rm split}$ is a spatial divergence, its volume integral is controlled by boundary fluxes. The model therefore naturally produces source--sink or bipolar patterns.

A genuine microscopic extra-current need not possess any of these restrictions: it is a four-dimensional scalar statement about the matter current, may have a non-zero local or integrated contribution depending on the microscopic process and boundary conditions, and does not require rotation. The present construction can mimic only those aspects of genuine non-conservation that enter the scalar field equation through a localized, transported, sign-changing source with the same symmetry content. In this precise sense the observer split is an effective surrogate for a limited class of extra-current profiles, not a microscopic realization of arbitrary charge non-conservation.

The restricted structure is also what makes the model falsifiable. A response attributed to the split source should inherit its oddness under rotation reversal and its suppression under restoration of axisymmetry. An effect that remains unchanged under $\Omega\to-\Omega$, or survives when $\partial_\phi\rho$ is eliminated while other experimental conditions are held fixed, would not support the present closure.
}

\section{What would be required for a true microscopic violation?}

Suppose one wants more than the effective bridge above, namely a genuine microscopic extra-current
$$
I = \nabla_\mu J^\mu \neq 0
$$
induced by rotation itself. Then the observer-split framework developed here is not sufficient. An additional physical mechanism is required.

Broadly speaking, three possibilities exist. First, intrinsically non-local matter dynamics, as in fractional or other non-local wave equations, may lead to a matter current that fails to satisfy a local continuity law \cite{Modanese2017dipole,MinottiModanese2023}. Second, quantum-coherent macroscopic matter sectors may exhibit an effective anomaly at the level of the current operator \cite{MinottiModanese2021,MinottiModanese2026}. Third, one may introduce a genuinely new coupling between the AB scalar sector and a gravitomagnetic invariant built from the vorticity or shift structure of the rotating congruence.

These possibilities lie beyond the present paper. Here we establish only an effective bridge based on the exact source term that appears in the rotating observer split of a covariantly conserved current.
 
\section{Operational signatures of the effective split-source model}

Because the proposed closure is observer-tied and effective, its empirical status should be assessed through operational signatures rather than through claims of fundamental charge violation. In the weak-field rotating regime, the model suggests the following possible checks:

{\color{black}
A useful feature of these signatures is that they can be converted into differential null tests. In particular, the combination of an odd response under rotation reversal with suppression for an axisymmetric source is more restrictive than a search for an unspecified rotation-dependent anomaly. A recently developed companion proposal makes this logic explicit for a small rotating aluminum disc carrying a fiber-optic ring, of which only a localized angular sector is optically active \cite{Iadicicco2026}. The illumination is used to modulate an effective surface or active-layer charge density in the metallic carrier; this induced density is to be measured or bounded independently by an electrometric or capacitive calibration synchronized with the optical modulation and rotor phase. The optical evanescent/tunnelling coefficient therefore serves only as a controllable gate and is not itself interpreted as charge density. With this operational definition the observer-split source retains the form $I_{\rm opt}\simeq\Omega\partial_\phi\rho_{\rm ind}$, and reversing the rotor direction changes its sign while nearly uniform illumination suppresses the angular-gradient term.
}

\begin{enumerate}

\item \emph{Rotation reversal}: Since $I_G$ is odd under $\beta \to -\beta$, any observable response that depends linearly on the split source changes sign under reversal of the sense of rotation.
\item \emph{Gradient dependence}: If the relevant charge density is spatially uniform, $I_G$ vanishes at leading order. The effect requires charge inhomogeneity.
\item \emph{Angular-structure dependence}: The observable response is controlled by non-axisymmetric charge gradients, especially transient ones.
\item \emph{Geometry dependence}: The observable response scales with the overlap between the charge structure and the shift or gravitomagnetic flow.
\item \emph{Transient sensitivity}: Smoothing or suppressing localized transient perturbations should reduce the source term.

\end{enumerate}

{\color{black}
The companion protocol proposes four especially direct controls at the source-generation stage: optical ON/OFF gating, CW/CCW reversal, comparison of a localized illuminated sector with nearly uniform $360^\circ$ illumination, and comparison of a conductive carrier with insulating or dummy controls \cite{Iadicicco2026}. For a localized angular profile of width $\sigma$, the characteristic source amplitude scales as $|I|_{\rm max}\propto |\Omega|/\sigma$, reproducing the dependence derived in Appendix~B.

In the present paper we retain only the electromagnetic consequence of this source construction. Once the optically modulated density has been independently calibrated, the proposed test is to search for the scalar/potential response associated with $S$ and to correlate it with the predicted source $I_{\rm opt}\simeq\Omega\partial_\phi\rho_{\rm ind}$. Detection strategies of the type previously discussed for gauge waves in extended AB electrodynamics provide a natural starting point \cite{MinottiModanese2023}. The decisive signature is differential: the electromagnetic response should reverse with $\Omega$, be strongly reduced when the illuminated density becomes axisymmetric, and scale with the calibrated angular gradient. A null result under a specified source amplitude and detector sensitivity would bound the observer-split scalar response assumed in Eq.~\eqref{eq:boxg}; a signal lacking the predicted reversal or axisymmetry dependence would not support the present closure.
}

A minimal reproducible numerical implementation of the transient model and its immediate diagnostics is provided in Appendix~D.

\section{Discussion and limitations}
The analysis developed in this paper leads to three main conclusions. First, standard general relativity coupled to standard charged matter does not predict a genuine microscopic violation of local charge conservation solely from a stationary gravitomagnetic background or from rotation. Second, a mathematically well-defined effective bridge nevertheless emerges after the $3+1$ split, where the projected continuity equation contains the exact split source term
$
I_{\mathrm{split}} = \frac{1}{N} D_i(\rho\,\beta^i),
$
whose weak-field limit is
$
I_G = D_i(\rho\,\beta^i).
$
Third, stronger microscopic claims require additional non-local, quantum-coherent, or explicitly new couplings beyond the framework developed here.

The present construction should therefore be interpreted as an effective phenomenological model. Its value lies in the fact that it is computable, falsifiable, and directly connected to experimentally accessible finite-scale signatures.

{\color{black}
The revised interpretation also narrows the meaning of ``frame dependence.'' The closure is not intended to generate a different physical law for every coordinate frame. Once the material congruence and experimental protocol are specified, they form part of the effective model, just as a medium four-velocity forms part of a relativistic constitutive description. The price is that the model is not fundamental or locally Lorentz-invariant in the same sense as the microscopic matter equations; the advantage is that its additional structure is physically identifiable rather than arbitrarily chosen. Whether this type of medium/congruence-dependent closure captures any real macroscopic response is therefore an experimental question.
}

\subsection{Local inertial removal and the scope of the effective claim}
A natural objection is that the split source term $I_G$
is local in form, and that the scalar sector is introduced through the local equation \eqref{eq:boxg}.
One may therefore ask whether the effect should disappear in a sufficiently small locally inertial frame, and if so why $I_G$ should be regarded as physically meaningful at all.

This objection is valid at the microscopic level and does not undermine the present construction. The underlying matter current remains covariantly conserved,
and this is precisely the content of the no-go result we have proven. Any sufficiently local freely falling description removes the apparent non-conservation associated with the observer-dependent split.

For this reason, $I_G$ should not be interpreted as a frame-independent local source law. Its meaning is operational: it is the exact source term that appears when charge transport is described in the rotating frame of a finite device. The proposed AB-type closure becomes meaningful only after coarse-graining over the spatial and temporal scales set by the apparatus, its transients, and the measurement protocol.

\section{Conclusion}
We have shown that rotating-frame geometry, together with stationary gravitomagnetic backgrounds when present, can be linked to extended Aharonov-Bohm electrodynamics only in a precise phenomenological sense. The point is not that rotating masses or rotating coordinates automatically violate local charge conservation. On the contrary, for standard generally covariant charged matter the physical four-current remains covariantly conserved, and this is the no-go result established here.

What rotation does provide, after a rotating $3+1$ decomposition, is an exact observer-dependent source term in the projected continuity equation.
This term is not a microscopic anomaly. It is the bookkeeping term that appears when the conserved four-current is rewritten in observer-measured transport variables adapted to a slicing with nonzero shift.

On that basis, we propose a phenomenological AB-type scalar-sector closure for finite rotating systems. In this restricted sense, rotation can be linked to extended Aharonov-Bohm electrodynamics without invoking anomalous tunnelling or a fundamental breakdown of charge conservation. The resulting framework is effective rather than fundamental, frame-tied rather than local-inertial, and experimentally meaningful only for finite rotating systems with non-uniform and transient charge structure.

{\color{black}
More precisely, the relevant dependence is on a physically selected material congruence rather than on an arbitrary coordinate representation. The construction represents only a restricted, spatially balanced class of effective extra-current profiles, with characteristic oddness under rotation reversal and suppression in the axisymmetric limit. These restrictions both distinguish the model from genuine frame-independent microscopic non-conservation and provide its sharpest falsification criteria. The recently proposed optically gated rotating-disc protocol illustrates how such criteria can be implemented in a finite laboratory system, while analogue-gravity media offer another possible direction in which the preferred congruence is supplied naturally by a material flow.
}

\appendix
\section{Derivation details for the projected source term $I_G$}
For completeness, we record the intermediate algebra leading from covariant conservation to the observer-split source term. Starting from
\begin{equation}
\nabla_\mu J^\mu = \frac{1}{\sqrt{-g}}\,\partial_\mu\!\left(\sqrt{-g}\,J^\mu\right)=0,
\qquad \sqrt{-g}=N\sqrt h,
\end{equation}
and introducing the Eulerian decomposition
\begin{equation}
J^\mu=\rho\,n^\mu+j^\mu,
\qquad
n_\mu j^\mu=0,
\qquad
n^\mu=\qty(\frac{1}{N},-\frac{\beta^i}{N}),
\end{equation}
so that
\begin{equation}
J^0=\frac{\rho}{N},
\qquad
J^i=j^i-\frac{\rho}{N}\beta^i,
\end{equation}
one obtains
\begin{align}
0 &= \partial_t\!\left(N\sqrt h\,J^0\right)+\partial_i\!\left(N\sqrt h\,J^i\right)\\
  &= \partial_t\!\left(\sqrt h\,\rho\right)+\partial_i\!\left(N\sqrt h\,j^i-\sqrt h\,\rho\,\beta^i\right).
\end{align}
Rearranging terms gives
\begin{equation}
\partial_t\!\left(\sqrt h\,\rho\right)+\partial_i\!\left(N\sqrt h\,j^i\right)=\partial_i\!\left(\sqrt h\,\rho\,\beta^i\right).
\end{equation}
Division by $\sqrt h$ and use of the spatial covariant derivative $D_i$ then yield
\begin{equation}
\partial_t\rho + D_i(Nj^i)=D_i(\rho\beta^i).
\end{equation}
In the weak-field regime $N\simeq 1$, so the projected continuity equation becomes
\begin{equation}
\partial_t\rho + D_i j^i = I_G,
\qquad
I_G\equiv D_i(\rho\beta^i).
\end{equation}
This makes clear that $I_G$ is not an ad hoc insertion: it is the exact source term produced by the $3+1$ split of a covariantly conserved current in a rotating background.

A simple dimensional check is also useful. Since $\beta^i$ has dimensions of velocity and $D_i$ contributes one inverse length, $I_G$ has the same dimensions as $\partial_t\rho$, namely charge density per unit time. This is the correct dimension for a source term in the projected continuity equation.

\section{Exact transient model and source suppression law}
To make the source condition explicit, consider a localized transient on a rotating ring,
\begin{equation}
\delta\rho_s(\phi,t)=A(t)\exp\!\left[-\frac{(\phi-\phi_c(t))^2}{2\sigma^2}\right],
\qquad \phi_c(t)=\Omega t+\phi_0.
\label{eq:appendixgauss}
\end{equation}
For rigid rotation about $z$, the effective source term is
\begin{equation}
I_G(\phi,t)=\Omega\,\partial_\phi\delta\rho_s
=-\Omega\,A(t)\,\frac{\phi-\phi_c(t)}{\sigma^2}
\exp\!\left[-\frac{(\phi-\phi_c(t))^2}{2\sigma^2}\right].
\label{eq:appendixIG}
\end{equation}
The maximum magnitude occurs at $|\phi-\phi_c|=\sigma$, hence
\begin{equation}
|I_G|_{\max}=\frac{|\Omega|\,A(t)}{\sigma}e^{-1/2}.
\label{eq:peakIG}
\end{equation}
Equation \eqref{eq:peakIG} makes the engineering dependence completely transparent: the source grows linearly with rotational rate and transient amplitude, and it grows when the transient becomes sharper in angle.

A mathematically clean model for transient suppression is convolution with a normalized Gaussian kernel of width $\lambda$. Convolving \eqref{eq:appendixgauss} with such a kernel leaves the profile Gaussian, but both broadens it and reduces its peak amplitude:
\begin{equation}
\delta\rho_s^{\rm supp}(\phi,t)
=
A(t)\frac{\sigma}{\sqrt{\sigma^2+\lambda^2}}
\exp\!\left[-\frac{(\phi-\phi_c(t))^2}{2(\sigma^2+\lambda^2)}\right].
\end{equation}
Therefore the suppressed effective source is
\begin{equation}
I_G^{\rm supp}(\phi,t)
=
-\Omega\,A(t)\,\frac{\sigma\,[\phi-\phi_c(t)]}{(\sigma^2+\lambda^2)^{3/2}}
\exp\!\left[-\frac{(\phi-\phi_c(t))^2}{2(\sigma^2+\lambda^2)}\right].
\end{equation}
The maximum magnitude again occurs at $|\phi-\phi_c|=\sqrt{\sigma^2+\lambda^2}$, hence
\begin{equation}
|I_G|_{\max}^{\rm supp}=|\Omega|\,A(t)\,e^{-1/2}\,\frac{\sigma}{\sigma^2+\lambda^2},
\end{equation}
and the exact reduction factor under normalized Gaussian smoothing is
\begin{equation}
\frac{|I_G|_{\max}^{\rm supp}}{|I_G|_{\max}}=\frac{\sigma^2}{\sigma^2+\lambda^2}\le 1.
\label{eq:suppressionlaw}
\end{equation}
Hence any low-pass system modeled in this way that smooths the transient, broadens its angular footprint, or lowers its amplitude must reduce the effective source term. If the transient suppression also attenuates amplitude by a factor $\eta\le 1$, then the total reduction becomes
\begin{equation}
\frac{|I_G|_{\max}^{\rm supp}}{|I_G|_{\max}}=\eta\,\frac{\sigma^2}{\sigma^2+\lambda^2}.
\end{equation}
This gives a compact explanation of why low-pass-type suppression should weaken the effective bridge when modeled by normalized Gaussian smoothing.

\section{Extension to a non-Gaussian waveform}\label{app:filtertheorem}
One can  prove a more general technical result that does not depend on the detailed waveform. Expand the transient charge perturbation on the rotating ring in angular Fourier modes,
\begin{equation}
\delta\rho_s(\phi,t)=\sum_{n\in\mathbb Z}\hat\rho_n(t)e^{in\phi}.
\end{equation}
Then, from $I_G=\Omega\,\partial_\phi\delta\rho_s$, one gets
\begin{equation}
I_G(\phi,t)=\Omega\sum_{n\in\mathbb Z} in\,\hat\rho_n(t)e^{in\phi}.
\end{equation}
Assume that a transient-suppression network acts linearly on the angular profile through modal multipliers $H_n$, so that
\begin{equation}
\widehat{\delta\rho_s^{\rm supp}}\!{}_n(t)=H_n\,\hat\rho_n(t),
\qquad |H_n|\le \eta\le 1.
\end{equation}
The suppressed effective source is then
\begin{equation}
I_G^{\rm supp}(\phi,t)=\Omega\sum_{n\in\mathbb Z} in\,H_n\hat\rho_n(t)e^{in\phi}.
\end{equation}
By Parseval's identity,
\begin{equation}
\|I_G^{\rm supp}(\cdot,t)\|_{L^2}^2
=2\pi\,\Omega^2\sum_{n\in\mathbb Z} n^2 |H_n|^2 |\hat\rho_n(t)|^2
\le \eta^2\,2\pi\,\Omega^2\sum_{n\in\mathbb Z} n^2 |\hat\rho_n(t)|^2
=\eta^2\|I_G(\cdot,t)\|_{L^2}^2.
\label{eq:L2bound}
\end{equation}
Hence
\begin{equation}
\|I_G^{\rm supp}(\cdot,t)\|_{L^2}\le \eta\,\|I_G(\cdot,t)\|_{L^2}.
\end{equation}
This gives a filter-agnostic theorem: \emph{any linear suppression network whose modal gains satisfy $|H_n|\le 1$ cannot increase the effective source norm, and any strict attenuation implies strict suppression of $I_G$ in the same norm}. If, moreover, the network is low-pass in the sense that $|H_n|$ decreases with $|n|$, then sharper transients are suppressed more strongly because the factor $n^2$ weights high-angular-frequency content more heavily in \eqref{eq:L2bound}.

This theorem is stronger than the Gaussian example because it does not depend on the detailed shape of the transient. It requires only linearity and bounded gain, which are the assumptions expected for a low-pass-type suppression. It therefore supports the operational claim that suppressing sharp transient structure suppresses $I_G$, without committing to a single waveform model.

\section{Minimal reproducible numerical model}
To keep the proposal falsifiable, we generated a minimal numerical model accompanying this manuscript. The Python script \texttt{gravitomagnetic\_ab\_model.py} and the MATLAB/Octave script \texttt{gravitomagnetic\_ab\_model.m} implement exactly the Gaussian transient model of Appendix B. They produce three immediate checks:
\begin{enumerate}[label=(\roman*)]
\item a localized transient generates a bipolar $I_G(\phi)$ profile according to \eqref{eq:appendixIG};
\item Gaussian smoothing reduces the peak source exactly as in \eqref{eq:suppressionlaw};
\item reversing the rotation $\Omega\to-\Omega$ flips the sign of $I_G$.
\end{enumerate}
These are consistency checks showing that the proposed effective source term is computable, reproducible, and experimentally falsifiable.

\begin{figure}[htbp]
\centering
\includegraphics[width=0.9\textwidth]{fig_transient_and_IG.png}
\caption{Localized transient $\delta\rho_s(\phi)$ and the corresponding induced source $I_G(\phi)=\Omega\partial_\phi\delta\rho_s$ on a rotating ring. Smoothing the transient lowers and broadens the induced source.}
\label{fig:transientIG}
\end{figure}

\begin{figure}[htbp]
\centering
\includegraphics[width=0.78\textwidth]{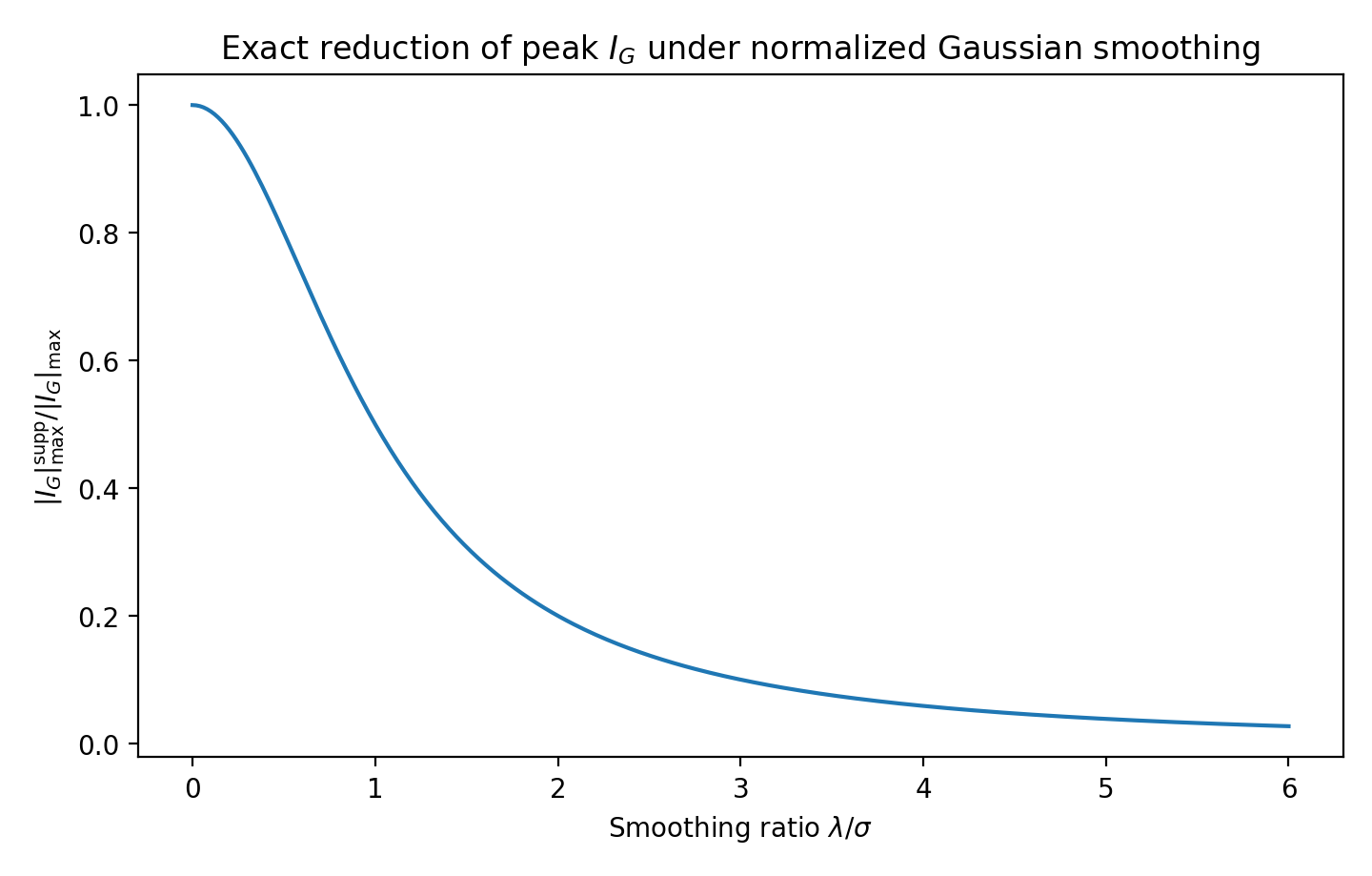}
\caption{Exact reduction of the peak effective source under normalized Gaussian smoothing. The normalized peak follows $\sigma^2/(\sigma^2+\lambda^2)$, so stronger smoothing necessarily suppresses $I_G$.}
\label{fig:suppressionlaw}
\end{figure}

\begin{figure}[htbp]
\centering
\includegraphics[width=0.78\textwidth]{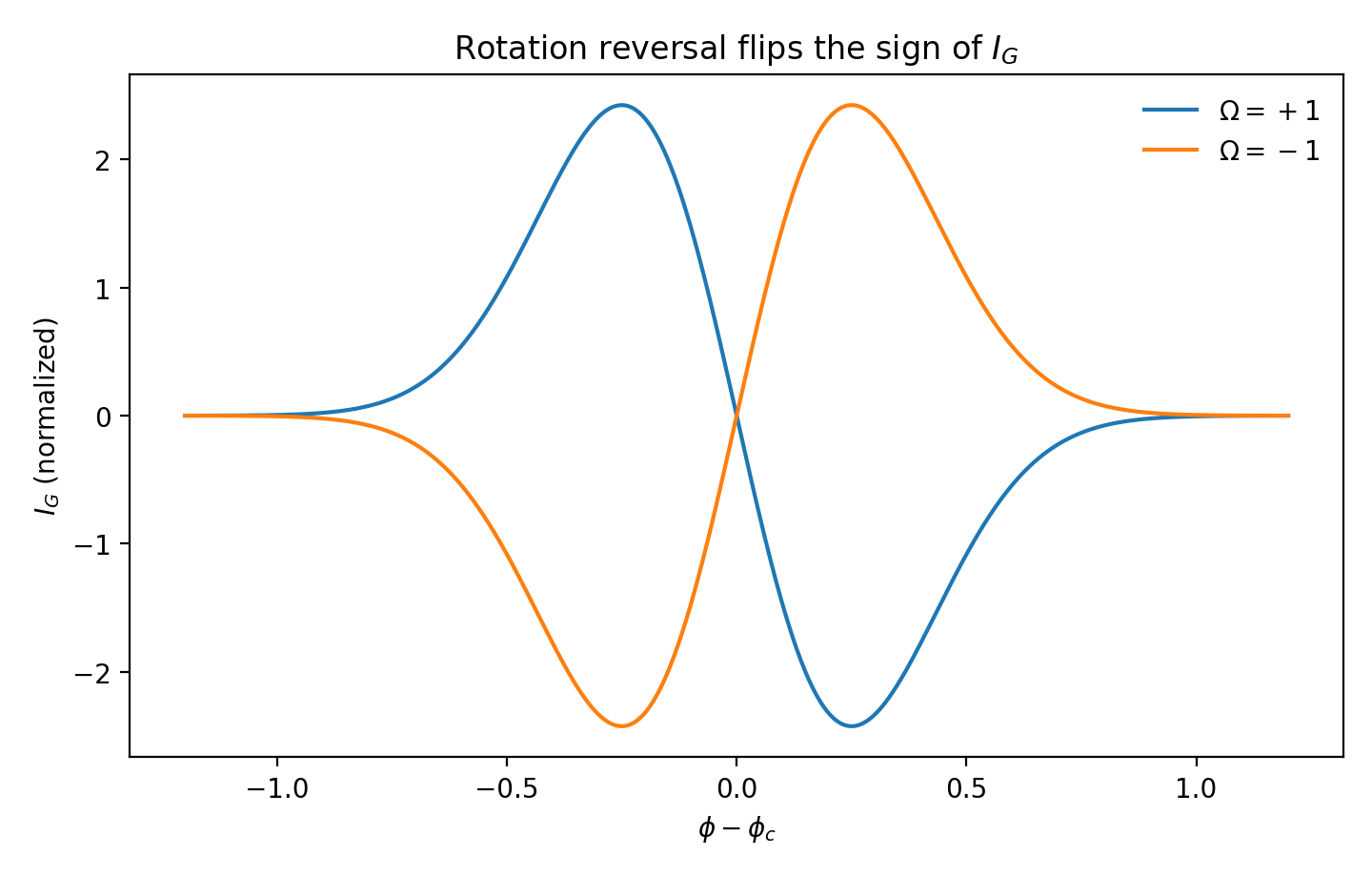}
\caption{Rotation reversal test. For the same localized transient, changing $\Omega\to-\Omega$ flips the sign of $I_G$, providing a sharp falsifiable prediction.}
\label{fig:signflip}
\end{figure}

\end{document}